\documentclass[letterpaper, 10 pt, conference, twocolumn]{IEEEtran}  

\IEEEoverridecommandlockouts                              

\usepackage{graphicx} %
\usepackage{amsmath} 
\usepackage{amssymb}  
\usepackage{eqnarray}  %
\usepackage{latexsym}  %

\title{\LARGE \bf
Tracking Control for FES-Cycling 
based on Force Direction Efficiency with 
Antagonistic Bi-Articular Muscles 
}

\author{Hiroyuki Kawai$^{1}$, Matthew J. Bellman$^{2}$, 
Ryan J. Downey$^{2}$ and Warren E. Dixon$^{2}$
\thanks{This research is supported in part by the Telecommunications 
Advancement Foundation, and NSF award number 1161260. 
Any opinions, findings and conclusions or recommendations expressed in 
this material are those of the authors and do not necessarily reflect 
the views of the sponsoring agency.}
\thanks{$^{1}$Department of Robotics, Kanazawa Institute of Technology, 
Ishikawa, 921-8151, Japan
        {\tt\small hiroyuki@neptune.kanazawa-it.ac.jp}}%
\thanks{$^{2}$Department of Mechanical and Aerospace Engineering, University of Florida, 
Gainesville, FL 32611-6250, USA
}%
}

\newcommand{\R}{{\cal R}}

\newcommand{\sgn}{\mathop{\mathrm{sgn}}\nolimits}

\newtheorem{theorem}{Theorem}
\newtheorem{assumption}{Assumption}

\begin{document}

\maketitle
\thispagestyle{empty}
\pagestyle{empty}

\begin{abstract}
A functional electrical stimulation (FES)-based 
tracking controller is developed to enable cycling 
based on a strategy to yield force direction efficiency 
by exploiting antagonistic bi-articular muscles. 
Given the input redundancy naturally occurring among 
multiple muscle groups, 
the force direction at the pedal is explicitly determined 
as a means to improve the efficiency of cycling. 
A model of a stationary cycle and rider is developed 
as a closed-chain mechanism. 
A strategy is then developed to switch between muscle 
groups for improved efficiency 
based on the force direction of each muscle group. 
Stability of the developed controller is analyzed through 
Lyapunov-based methods. 
\end{abstract}

\section{INTRODUCTION}

In the human body, coordinated firing of motor neurons activates 
skeletal muscles which generate torques about the body's joints, and 
thereby, produce complex motions. Neurological disorders that damage 
the motor neurons can lead to paresis or paralysis and impaired motion. 
Specifically, people suffering from upper motor neuron 
disorders like stroke and spinal cord 
injury have difficulty performing functional motions like standing, walking, 
or cycling. 
Functional electrical stimulation (FES) seeks to augment
lost motor neuron function through artificially applied electrical 
currents to recover some 
functional motion (e.g., walking \cite{TBE12NeEr}, 
standing \cite{TNSRE05KaShMuBa}, grasping and releasing 
\cite{TNSRE12WeScVeKo}, etc.).  

Cycling induced by FES has been reported as physiologically and 
psychologically beneficial for people suffering from disorders 
affecting the muscles of the lower limbs \cite{JMBE11PeChLaChChMiHa}; 
however, FES-cycling is metabolically inefficient and produces 
less power output than able-bodied cycling \cite{THC12HuFaSaGrLa}. 
Previous studies have used various design and control techniques 
to address these shortcomings. 
Chen $et\ al.$ \cite{TRE97ChYuHuAnCh} used a model-free fuzzy 
logic controller for FES-cycling.  
Gf${\rm \ddot{o}}$ohler and Lugner \cite{TRE00GfLu} considered 
an optimized stimulation pattern of leg muscles by FES. 
In \cite{TNSRE04GfLu}, the influence of a number of individual parameters
on the optimal stimulation pattern and power output during
FES-cycling was investigated.  
Eser $et\ al.$ \cite{TNSRE03EsDoKnSt} examined 
the relation between stimulation frequency and power output for cycling 
by trained SCI patients. 
Kim $et\ al.$ \cite{SMC08KiEoHaKhTaYiJu} proposed a feedback
control system for FES-cycling, focusing on automatically determining 
stimulation patterns for multiple muscle groups. 

The aforementioned results provide useful methods for FES-cycling 
from a practical perspective,  
though explicit analysis of FES-cycling control from a theoretical 
point of view has been limited to linear approximations of the 
nonlinear cycle-rider system. 
Some recent studies \cite{TNSRE09ShStGrDi}--\cite{CDC13DoChDi} have 
focused on the development of RISE-based FES controllers and the 
associated analytical stability analysis for tracking of a human 
knee joint in the presence of a nonlinear uncertain muscle model 
with non-vanishing additive disturbances. However, these previous 
works have only considered knee joint dynamics.

In this paper, we consider tracking control for FES-cycling 
based on force direction efficiency derived from using 
antagonistic bi-articular muscles. 
Antagonistic bi-articular muscles, which pass over two adjacent joints 
and therefore act on the both joints simultaneously, are considered 
as one of the most important mechanisms of the human body associated 
with motion \cite{JDSMC85Ho}. 
Based on the antagonistic bi-articular muscle model, a stimulation 
pattern is derived for the gluteal, quadriceps femoris, hamstrings, 
and gastrocnemius muscle groups which aims to improve efficiency by 
maintaining a pedal force that is tangent to the pedal path. 
Considering the bi-articular muscle effects and controlling the pedal 
force direction may prove to increase FES-cycling power output and 
efficiency. A RISE-based controller and associated stability analysis 
are developed for an uncertain nonlinear cycle-rider system in the 
presence of an unknown time-varying disturbance, and semi-global 
asymptotic tracking of the desired trajectories is guaranteed 
provided sufficient control gain conditions are satisfied.

\section{BICYCLE MODEL}
A stationary cycle and rider can be modeled as a closed-chain 
mechanism~\cite{TCST00GhChGuLo}. 
Consider a three degree-of-freedom (DOF) holonomic mechanical multibody 
system $\Sigma'$ as shown in Fig.~\ref{fig:BicycleModel2}, which 
consists of a collection of rigid bodies described as 
\begin{eqnarray}
\Sigma':\ M'(q')\ddot{q}'+C'(q',\dot{q}')\dot{q}'+g'(q')=0, 
\label{eqn:Sigma'}
\end{eqnarray}
where $q'=[q_{1}\ q_{2}\ q_{3}]^{T} \in \R^{3}$ 
represents the hip, knee, and crank angles, respectively, 
$M'(q')\in \R^{3\times 3}$ is the inertia matrix, 
$C'(q',\dot{q}')\dot{q}'\in \R^{3}$ represents the 
centrifugal and Coriolis terms, and $g'(q')\in \R^{3}$ is the 
gravity term.

\begin{figure}[t]
 \begin{center}
  \scalebox{.48}{\includegraphics{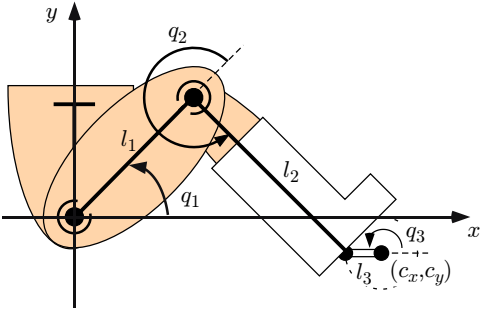}}
  \caption{Bicycle Model.}
  \label{fig:BicycleModel2}
 \end{center}
\end{figure}
From Fig.~\ref{fig:BicycleModel2}, the scleronomic holonomic constraints 
are given by 
\begin{eqnarray}
{\cal C}:\ \phi(q')=
\left[\begin{array}{c}
l_{1}C_{1}+l_{2}C_{12}-l_{3}C_{3}-c_{x}\\
l_{1}S_{1}+l_{2}S_{12}-l_{3}S_{3}-c_{y}
\end{array}\right]
=0
\label{eqn:constraints}
\end{eqnarray}
where $l_{i}$ $(i=1,2,3)$ are the lengths of the thigh, shank and crank, 
$c_{x}$ and $c_{y}$ are the coordinates of the 
center of the crank,  
$S_{ij}$ is defined as $S_{ij}:=\sin(q_{i}+q_{j})$, and 
$C_{ij}$ is defined as $C_{ij}:=\cos(q_{i}+q_{j})$. 
\begin{assumption}\label{ass:1}
From (\ref{eqn:constraints}) and the physical relationships 
associated with the seated cyclist, 
the hip and knee angles are constrained to the regions 
$\pi<q_{2}<2\pi$ and $\pi<q_{1}+q_{2}<2\pi$. 
\end{assumption}

In the subsequent development, only the crank angle $q_{3}$ is 
assumed to be measurable.  
Hence, a parameterization for the generalized coordinates $q$ 
is developed as
\begin{eqnarray}
q'\longmapsto q=\alpha(q')=\left[\begin{array}{ccc}
0&0&1
\end{array}\right]q'.\label{eqn:parameterization}
\end{eqnarray}
From Theorem 1 in \cite{TCST00GhChGuLo}, the equation of motion of 
the constrained system expressed in terms of the independent generalized 
coordinate $q$ is obtained by combining
\begin{eqnarray}
\left\{\begin{array}{l}
M(q')\ddot{q}+C(q',\dot{q}')\dot{q}+g(q')=0 \cr
\dot{q}'=\mu(q')\dot{q} \cr
q'=\sigma(q)
\end{array}\right. 
\end{eqnarray}
to yield
\begin{eqnarray}
 \Sigma:\ M(q)\ddot{q}+C(q,\dot{q})\dot{q}+g(q)=\tau \label{eqn:Sigma}
\end{eqnarray}
where $\tau\in \R$ is the torque about the crank,  
$\mu(q')$ is expressed by using the constraints in (\ref{eqn:constraints})
and the parameterization in (\ref{eqn:parameterization}), and 
$\sigma(q)$ can be derived by solving the constraints ${\cal C}$ in
(\ref{eqn:constraints}). 
Detailed expressions for $\mu(q')$ and $\sigma(q)$ are in 
Appendix~A.

\section{INPUT FORCE}
The human thigh model can be divided into three pairs of antagonistic
muscles as depicted in Fig.~\ref{fig:human_leg}, where 
two groups consist of antagonistic mono-articular muscles 
and one group consists of antagonistic bi-articular muscles.  
The antagonistic mono-articular muscles that span the hip joint consist of 
three extensor muscles denoted by $e_{1}$ and two flexor muscles 
denoted by $f_{1}$. 
The antagonistic mono-articular muscles that span the knee joint consist 
of a flexor muscle denoted by $f_{2}$ and three extensor muscles 
denoted by $e_{2}$. 
Antagonistic bi-articular muscles span both the hip 
and the knee joint and consist of $fe_{3}$ and $ef_{3}$  
where $fe_{3}$ flexes the hip and extends the knee, 
and $ef_{3}$ extends the hip and flexes the knee.  
\begin{figure}[t]
 \begin{center}
  \scalebox{.48}{\includegraphics{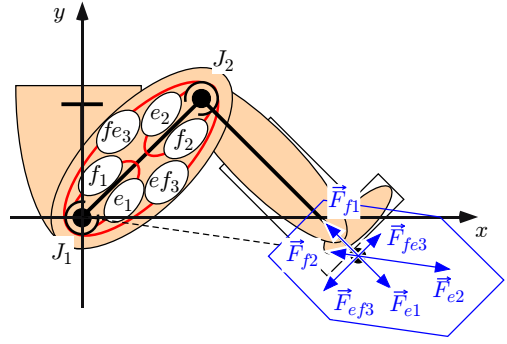}}
  \caption{Human thigh model. 
(i) Antagonistic mono-articular muscles 
spanning the hip joint consist of three extensor muscles $e_{1}$, 
i.e.,~gluteus maximus, gluteus medius and gluteus minimus, 
and two flexor muscles $f_{1}$, i.e.,~psoas major and iliacus.
(ii) Antagonistic mono-articular muscles spanning the knee joint 
consist of biceps femoris short head $f_{2}$ and three extensor 
muscles $e_{2}$, 
i.e.,~vastus intermedius, vastus lateralis and vastus medialis. 
(iii) Antagonistic bi-articular muscles spanning both the hip 
and the knee joint consist of rectus femoris $fe_{3}$ and  
three muscles $ef_{3}$, 
i.e.,~biceps femoris long head, semimembranosus and semitendinosus.  
$fe_{3}$ flexes the hip and extends the knee, 
and $ef_{3}$ extends the hip and flexes the knee.  
}
  \label{fig:human_leg}
 \end{center}
\end{figure}

The resulting force at the pedal depends on 
the combination of the active muscle forces.  
Moreover, as shown in Fig.~\ref{fig:human_leg}, 
the directions of $\vec{F}_{f1}$ and $\vec{F}_{e1}$ coincide 
with the direction of the shank, 
the direction of $\vec{F}_{f2}$ and $\vec{F}_{e2}$ pass through the hip 
joint $J_{1}$ and the pedal, 
and the directions of $\vec{F}_{fe3}$ and $\vec{F}_{ef3}$ are parallel to 
the thigh. 
\begin{figure}[t]
 \begin{center}
  \scalebox{.48}{\includegraphics{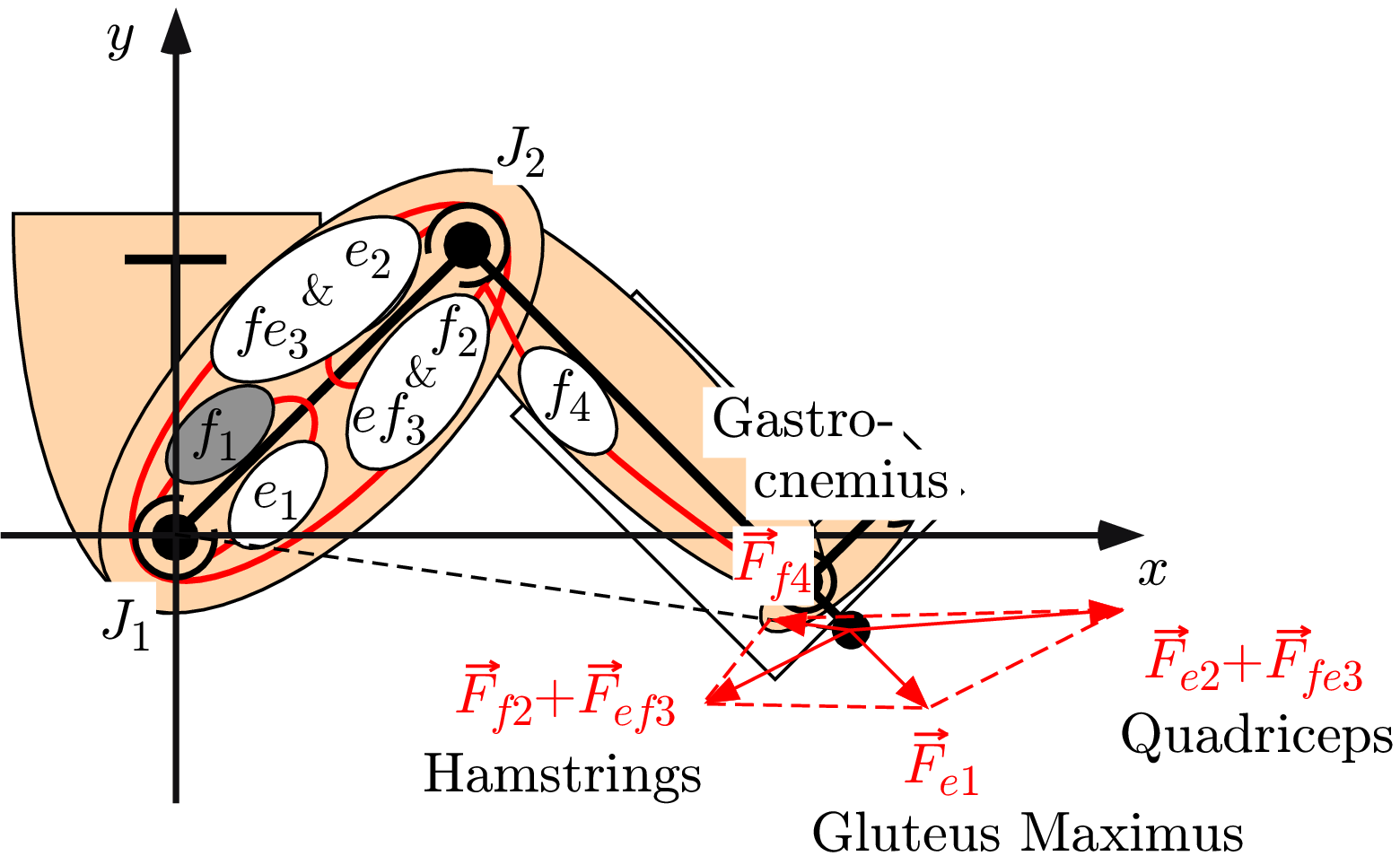}}
  \caption{
The gastrocnemius is a bi-articular muscle group which 
which is a flexor for the knee joint and 
is denoted as $f_{4}$.  
}
  \label{fig:human_leg+f4}
 \end{center}
\vspace{2ex}
 \begin{center}
  \scalebox{.58}{\includegraphics{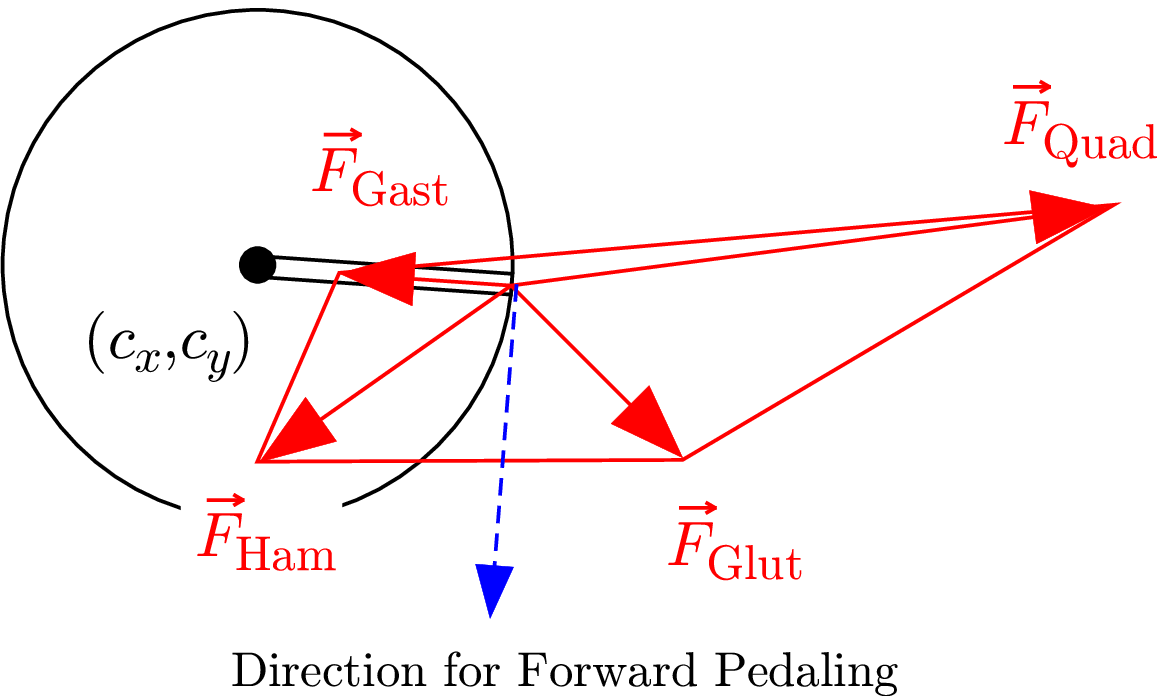}}
  \caption{To improve pedaling efficiency, the force direction
at the pedal is selected to be tangent to the pedal path by altering 
the relative activation of muscle groups.}
  \label{fig:force_direction}
 \end{center}
\end{figure}
The torque produced at the joint(s) the muscle spans is defined as
\begin{eqnarray}
\tau_{i}&:=&\Omega_{i}u_{i},\ \ 
\Omega_{i}:=\zeta_{i}\eta_{i}\cos(a_{i}),\label{eqn:each_force}\\
&&i\in {\cal T},\ \ {\cal T}:=\{e_{1}, f_{1}, e_{2}, f_{2}, 
ef_{3}, fe_{3}\},\nonumber
\end{eqnarray}
where $\zeta_{i}\in\R$  denotes a positive moment arm that changes
with the crank angle \cite{JBio04KrPaPe}, \cite{RehEng97Buetal},  
$a_{i}\in\R$ is defined as the pennation angle between the 
tendon and the muscle which changes with the crank angle 
\cite{TNSRE09ShStGrDi}, 
$\eta_{i}\in\R$ is an unknown function that relates the applied 
voltage to muscle fiber force which changes with the crank angle 
and velocity, and $u_{i}\in\R$ is the control voltage input applied 
across each muscle group.
\begin{assumption}\label{ass:2}
The moment arm $\zeta_{i}$ is assumed to be a positive, bounded, 
second order differentiable function such that 
its first and second time derivatives are bounded 
if $q^{k}\in{\cal L}_{\infty}$, where $q^{k}$ denotes the $k$th
time derivative of $q$ for $k=0$, $1$, $2$ \cite{JBio04KrPaPe}. 
Similarly, the function $\eta_{i}$ is assumed to be a positive, bounded, 
second order differentiable function such that 
its second time derivative is bounded if $q^{k}\in{\cal L}_{\infty}$ 
for $k=0$, $1$, $2$, $3$ \cite{RehEng99WaFuHo}. 
\end{assumption}
\begin{assumption}\label{ass:3}
For each bi-articular muscle, the torque acting on each of 
the two joints is assumed to be equal. 
\end{assumption}
The forces at the pedal $F=[F_{x}\ F_{y}]^{T}$ 
are related to the joint torque $T=[T_{1}\ T_{2}]^{T}$ as 
\begin{eqnarray}
 F=(J^{T})^{-1}T\label{eqn:FT}
\end{eqnarray}
where $J$ is the Jacobian matrix\footnote{
$\det(J^{T})=l_{1}l_{2}S_{2}\neq 0$ except for $q_{2}=n\pi,\ n\in {\cal Z}$. 
Thus, $J^{T}$ is invertible under Assumption~\ref{ass:1}.}
defined as
\begin{eqnarray}
J:=\left[\begin{array}{cc}
-l_{1}S_{1}-l_{2}S_{12}&-l_{2}S_{12}\cr
l_{1}C_{1}+l_{2}C_{12}&l_{2}C_{12}\cr
\end{array}\right]. \label{eqn:Jacobian}
\end{eqnarray}
Moreover, the joint torques can be represented as 
\begin{eqnarray}
 T_{1}&=&(\tau_{f1}-\tau_{e1})+(\tau_{fe3}-\tau_{ef3}),\label{eqn:T1}\\
 T_{2}&=&(\tau_{e2}-\tau_{f2})+(\tau_{fe3}-\tau_{ef3}).\label{eqn:T2}
\end{eqnarray}
Using (\ref{eqn:FT})--(\ref{eqn:T2}), 
the force at the pedal can be expressed as follows
\cite{IGAKU08NaKu}:   
\begin{eqnarray}
|F_{i}|&=&\sqrt{F_{i_{x}}^{2}+F_{i_{y}}^{2}}
=R_{i}\tau_{i},\\ 
\theta_{i}&=&\tan^{-1}\left(
\dfrac{F_{i_{y}}}{F_{i_{x}}}
\right),
\ \ i\in {\cal T}, 
\end{eqnarray}
where 
\begin{eqnarray}
R_{f1}&=&\left|\frac{1}{l_{1}S_{2}}\right|,\  \theta_{f1}=q_{1}+q_{2}-\pi, \\
R_{e1}&=&\left|\frac{1}{l_{1}S_{2}}\right|,\  \theta_{e1}=q_{1}+q_{2}, 
%
\end{eqnarray}
\begin{eqnarray}
R_{f2}&=&\left|\frac{1}{l_{2}S_{0}}\right|,\  
\theta_{f2}=\tan^{-1}\left(\dfrac{1_{1}S_{1}+l_{2}S_{12}}
{1_{1}C_{1}+l_{2}C_{12}}
\right)-\pi,\hspace*{3ex} \label{eqn:Rf2}
\\
R_{e2}&=&\left|\frac{1}{l_{2}S_{0}}\right|,\  
\theta_{e2}=\tan^{-1}\left(\dfrac{1_{1}S_{1}+l_{2}S_{12}}
{1_{1}C_{1}+l_{2}C_{12}}
\right), \label{eqn:Re2} 
\\
%
R_{fe3}&=&\left|\frac{1}{l_{2}S_{2}}\right|,\  
\theta_{fe3}=q_{1},\\
R_{ef3}&=&\left|\frac{1}{l_{2}S_{2}}\right|,\  
\theta_{ef3}=q_{1}-\pi, 
\end{eqnarray}
and $\theta_{i}$ is the direction of the 
force at the pedal. 
Note that (\ref{eqn:Rf2}) and (\ref{eqn:Re2}) make 
use of the geometric relation 
$S_{0}=-{l_{1}S_{2}}/{\sqrt{l_{1}^2+l_{2}^{2}+2l_{1}l_{2}C_{2}}}$ 
where $q_{0}:=2\pi-(q_{1}+q_{2})+\tan^{-1}((l_{1}S_{1}+l_{2}S_{12})
/(l_{1}C_{1}+l_{2}C_{12}))$. 

While healthy individuals may be able to activate individual 
muscles during voluntary contractions, it is difficult to
selectively activate individual muscles during external FES 
with transcutaneous electrodes if the muscles are in close 
proximity to each other (e.g., the pair of vastus intermedius, 
vastus lateralis and vastus medialis $e_{2}$ 
and rectus femoris $fe_{3}$,  
the pair of biceps femoris short head $f_{2}$ and 
biceps femoris long head, semimembranosus and semitendinosus   
$fe_{3}$.). 
Moreover, deep muscles (e.g., psoas major and iliacus $f_{1}$) 
cannot be activated by transcutaneous stimulation without also
activating the superficial muscles. 
Therefore, we consider the quadriceps femoris muscle group 
which contains $e_{2}$ and $fe_{3}$, and 
the hamstrings muscle group which contains 
$f_{2}$ and $ef_{3}$ as shown in Fig.~\ref{fig:human_leg+f4}. 
Additionally, the gastrocnemius, $f_{4}$, is a flexor muscle for the 
knee joint and is used to modify the direction of force.  
Hereafter, we consider the following four muscle groups: 
Gluteus Maximus, Hamstrings, Gastrocnemius, and Quadriceps. 

The forces acting at the pedal for each muscle group 
are expressed as 
\begin{eqnarray}
 \vec{F}_{\rm Glut}&=&\vec{F}_{e1},\\
 \vec{F}_{\rm Ham}&=&\vec{F}_{f2}+\vec{F}_{ef3},\\
 \vec{F}_{\rm Gast}&=&\vec{F}_{f4},\\
 \vec{F}_{\rm Quad}&=&\vec{F}_{e2}+\vec{F}_{fe3}, 
\end{eqnarray}
where $\vec{F}_{f4}$ is similar to $\vec{F}_{f2}$ 
(i.e., $R_{f4}=R_{f2}$ and $\theta_{f4}=\theta_{f2}$). 
The crank torque can be expressed in terms of the muscle forces as 
\begin{eqnarray}
\tau&=&(\vec{F}_{\rm Glut}+\vec{F}_{\rm Ham}+\vec{F}_{\rm Gast}
+\vec{F}_{\rm Quad})\times \vec{l}_{3}
\nonumber\\
&&-d+M_{e}(q)+M_{v}(\dot{q})\label{eqn:tau}
\end{eqnarray}
where 
$M_{e}(q)\in\R$ and $M_{v}(\dot{q})\in\R$ are elastic \cite{TRE00FePe} 
and viscous moments \cite{CEP05ScNePrHuFrFeRa} 
defined as
\begin{eqnarray}
M_{e}(q)
&:=&\mu(q')^{T}
\left[\begin{array}{c}
-k_{11}e^{-k_{12}q_{1}}(q_{1}-k_{13})\cr
-k_{21}e^{-k_{22}q_{2}}(q_{2}-k_{23})\cr
0\cr
\end{array}\right]\\
M_{v}(\dot{q})
&:=&\mu(q')^{T}
\left[\begin{array}{c}
b_{11}\tanh(-b_{12}\dot{q}_{1})-b_{13}\dot{q}_{1}\cr
b_{21}\tanh(-b_{22}\dot{q}_{2})-b_{23}\dot{q}_{2}\cr
0\cr
\end{array}\right],  
\end{eqnarray}
where $k_{11}$, $\cdots$, $k_{23}\in\R$ and 
$b_{11}$, $\cdots$, $b_{23}\in\R$ are unknown constants, 
$\vec{l}_{3}$ is defined as
\begin{eqnarray}
\vec{l}_{3}=l_{3}\left[\begin{array}{c}
C_{3}\cr
S_{3}
\end{array}\right], 
\end{eqnarray}
and $d$ is an unknown bounded disturbance from unmodeled dynamics. 
Combining (\ref{eqn:Sigma}) and (\ref{eqn:tau}) yields 
\begin{eqnarray}
\hspace*{-2ex}
M(q)\ddot{q}&+&C(q,\dot{q})\dot{q}+g(q) \nonumber\\
&=&\left(\sum_{i\in {\cal S}}\vec{\Omega}_{i}u_{i}\times \vec{l}_{3}\right)
-d+M_{e}(q)+M_{v}(\dot{q})\label{eqn:OpenDyn1}
\end{eqnarray}
where ${\cal S}=\{\rm Glut, Ham, Gast, Quad\}$ and 
\begin{eqnarray}
\vec{\Omega}_{\rm Glut}&:=&R_{e1}\Omega_{e1}
\left[\begin{array}{c}
C_{12}\cr
S_{12}\cr
\end{array}\right]\label{eqn:Omega_Glut}\\
\vec{\Omega}_{\rm Ham}&:=&
R_{f2}\Omega_{f2}
\left[\begin{array}{c}
C_{\theta_{f2}}\cr
S_{\theta_{f2}}\cr
\end{array}\right]
-R_{ef3}\Omega_{ef3}
\left[\begin{array}{c}
C_{1}\cr
S_{1}\cr
\end{array}\right]\label{eqn:Omega_Ham}\\
\vec{\Omega}_{\rm Gast}&:=&R_{f4}\Omega_{f4}
\left[\begin{array}{c}
C_{\theta_{f2}}\cr
S_{\theta_{f2}}\cr
\end{array}\right]\label{eqn:Omega_Gast}\\
\vec{\Omega}_{\rm Quad}&:=&
R_{e2}\Omega_{e2}
\left[\begin{array}{c}
C_{\theta_{e2}}\cr
S_{\theta_{e2}}\cr
\end{array}\right]
+R_{fe3}\Omega_{fe3}
\left[\begin{array}{c}
C_{1}\cr
S_{1}\cr
\end{array}\right]. \label{eqn:Omega_Quad}
\end{eqnarray}

Given the natural muscle redundancy, a transformation is 
developed as 
\begin{eqnarray}
u_{i}=\chi_{i}u,\ \ i\in {\cal S} \label{eqn:u_i}
\end{eqnarray}
where $u\in\R$ is the control input, 
and $\chi_{i}\in [0,1]$ is the designed 
activation ratio used to control force direction. 
The position of the pedal exists inside of the quadrilateral which is
constructed by the force directions of the four muscle groups as 
shown in Fig.~\ref{fig:force_direction}, and thus, the resulting 
force can be selected to be in any direction by altering the 
relative activation of the muscle groups. 
Because there exists infinitely many combinations by which three 
or more muscle groups can result in the same desired force 
direction, only two muscle groups are activated at any 
given time in this approach. 
The designed activation ratios are selected to satisfy 
the following relationships:
\begin{eqnarray}
\chi_{i}+\chi_{j}=1,\ \chi_{k}=0,\ \chi_{l}=0,\  \sin\theta = 1
\label{eqn:constraint_chi}
\end{eqnarray} 
where  $(i,j)\in\{({\rm Glut, Ham}),({\rm Ham, Gast}), ({\rm Gast, Quad)}$, 
$({\rm Quad, Glut})\}$ and $(k,l)\in {\cal S}\neq i,j$, and   
$\theta$ is the angle between the direction of the combination of the 
muscle forces  $\sum_{i\in {\cal S}}\chi_{i}\vec{\Omega}_{i}$ 
and the crank $\vec{l}_{3}$. 
The constraint on $\theta$ in (\ref{eqn:constraint_chi}) is designed 
such that the resulting combination of muscle forces is tangent to 
the crank. In other words, $\chi_{i}$ is designed to improve the 
efficiency of the cycling by ensuring that the resulting combination 
of muscle forces only contributes to the forward movement of the crank. 
By using (\ref{eqn:u_i}), (\ref{eqn:OpenDyn1}) can be expressed as 
\begin{eqnarray}
M(q)\ddot{q}&+&C(q,\dot{q})\dot{q}+g(q)\nonumber\\
&&-M_{e}(q)-M_{v}(\dot{q})+d
=\Omega_{\chi}u 
\label{eqn:OpenDyn2}
\end{eqnarray}
where $\Omega_{\chi}=\left\|
\sum_{i\in{\cal S}}\chi_{i}\vec{\Omega}_{i}\right\|l_{3}$. 

To design $\chi_{i}$ and satisfy the constraint on 
$\theta$ in (\ref{eqn:constraint_chi}), the magnitude and direction 
must be known for $\vec{\Omega}_{\rm Glut}$, $\vec{\Omega}_{\rm Ham}$,  
$\vec{\Omega}_{\rm Gast}$, and $\vec{\Omega}_{\rm Quad}$.   
The directions of $\vec{\Omega}_{\rm Glut}$ and $\vec{\Omega}_{\rm Gast}$ 
can be obtained analytically as a function of the crank 
angle\footnote{
Analytic solutions of $q_{3}$ at $\chi_{{\rm Glut}}=1$ and 
$\chi_{{\rm Gast}}=1$ are shown in Appendix~B.}. 
However, $\vec{\Omega}_{\rm Ham}$ and $\vec{\Omega}_{\rm Quad}$ consist 
of multiple muscles where the force directions are known but the relative 
magnitudes of the forces are unknown, and thus, the directions of 
$\vec{\Omega}_{\rm Ham}$ and $\vec{\Omega}_{\rm Quad}$ have to be 
estimated numerically from experimental data. Further, the relative
magnitudes of $\vec{\Omega}_{\rm Glut}$, $\vec{\Omega}_{\rm Ham}$,  
$\vec{\Omega}_{\rm Gast}$, and $\vec{\Omega}_{\rm Quad}$ are 
unknown functions of the crank angle and crank velocity, and thus, 
the activation ratio $\chi_{i}$ must be designed based on 
experimental data. 
\begin{assumption}\label{ass:4}
The first and second partial derivatives of $\chi_{i}$ with respect 
to the crank angle and crank velocity are assumed to exist 
and are bounded. 
Thus from (\ref{eqn:Omega_Glut})--(\ref{eqn:Omega_Quad}) 
and Assumption \ref{ass:2}, 
the first and second partial derivatives of $\Omega_{\chi}$
are bounded if $q^{k}\in{\cal L}_{\infty}$ 
for $k=0$, $1$, $2$, $3$.
and $\Omega_{\chi}$ is assumed to be a bounded 
function. 
\end{assumption}
From Assumption \ref{ass:2}, 
$\Omega_{i}$, $i\in {\cal T'}$, ${\cal T'}:=\{e_{1}$, $e_{2}$, $f_{2}$, 
$ef_{3}$, $fe_{3}$, $f_{4}\}$ is bounded such that 
$\xi_{i}>\Omega_{i}>\varepsilon_{i}>0$, $i\in {\cal T'}$ 
where $\xi_{i}$, $\varepsilon_{i}\in\R$ 
are positive constants. 
Further, from Assumptions \ref{ass:1} and \ref{ass:2}, 
$\Omega_{\chi}$ is bounded such that 
$\xi_{\Omega_{\chi}}>\Omega_{\chi}>\varepsilon_{\Omega_{\chi}}>0$  
where $\xi_{\Omega_{\chi}}$, $\varepsilon_{\Omega_{\chi}}\in \R$ are 
positive constants. 

\section{STABILITY ANALYSIS}\label{sec:stability}
The position error is defined as
\begin{eqnarray}
 e_{1}&=&q_{d}-q\label{eqn:e1}
\end{eqnarray}
where $q_{d}$ is the desired crank angle which is designed such that 
$q_{d}$, $q_{d}^{k}\in {\cal L}_{\infty}$, where $q_{d}^{k}$ denotes 
the $k$th time derivative of $q_{d}$ for $k=1,2,3,4$. 
To facilitate the subsequent analysis, 
the filtered tracking errors $e_{2}$,  $r\in\R$ are defined as
\begin{eqnarray}
 e_{2}&=&\dot{e}_{1}+\alpha_{1}e_{1}\label{eqn:e2}\\ 
 r&=&\dot{e}_{2}+\alpha_{2}e_{2}\label{eqn:r}
\end{eqnarray}
where $\alpha_{1}$, $\alpha_{2}\in\R$ are selectable 
positive constants. 
By using (\ref{eqn:e1})--(\ref{eqn:r}), 
the crank dynamics in (\ref{eqn:OpenDyn2}) can be transformed as follows
\begin{eqnarray}
M(q)r
&=&M(q)(\ddot{q}_{d}+\alpha_{1}\dot{e}_{1}+\alpha_{2}e_{2})
+C(q,\dot{q})\dot{q}\nonumber\\
&&\ \ \ -M_{e}(q)-M_{v}(\dot{q})
+g(q)+d-\Omega_{\chi}u\nonumber\\
&=&W+d-\Omega_{\chi}u \label{eqn:OpenDyn3}
\end{eqnarray}
where $W$ is defined as
\begin{eqnarray}
W&:=&M(q)(\ddot{q}_{d}+\alpha_{1}\dot{e}_{1}+\alpha_{2}e_{2})
+C(q,\dot{q})\dot{q}\nonumber\\
&&\ \ \ -M_{e}(q)-M_{v}(\dot{q})+g(q). 
\end{eqnarray}
After multiplying (\ref{eqn:OpenDyn3}) by $\Omega_{\chi}^{-1}$, 
the following dynamics can be obtained. 
\begin{eqnarray}
M_{\Omega}(q,\dot{q})r&=&W_{\Omega}-u+d_{\Omega} \label{eqn:OpenDyn4}
\end{eqnarray}
where $M_{\Omega}(q,\dot{q})$, $W_{\Omega}$ and $d_{\Omega}$ 
are defined as
\begin{eqnarray*}
 M_{\Omega}(q,\dot{q})&:=&\Omega_{\chi}^{-1}M(q) \\
 W_{\Omega}&:=&\Omega_{\chi}^{-1}W
\nonumber\\
&=&M_{\Omega}(q,\dot{q})(\ddot{q}_{d}+\alpha_{1}\dot{e}_{1}+\alpha_{2}e_{2})
+C_{\Omega}(q,\dot{q})\dot{q}\nonumber\\
&&\hspace{5ex}
-M_{e\Omega}(q,\dot{q})-M_{v\Omega}(q,\dot{q})+g_{\Omega}(q,\dot{q})\\
d_{\Omega}&:=&\Omega_{\chi}^{-1}d.
\end{eqnarray*}
From Assumptions~\ref{ass:1}, \ref{ass:2}, \ref{ass:4} and the facts that 
$\underline{M}\leq M(q)\leq\overline{M}$ 
where $\underline{M}$ and $\overline{M}$ are positive constants, 
we have that
\begin{eqnarray}
\underline{M}_{\Omega} \leq M_{\Omega}\leq \overline{M}_{\Omega}, 
\end{eqnarray}
where $\underline{M}_{\Omega}$, $\overline{M}_{\Omega}\in\R$ are 
positive constants. 
Also, the following auxiliary terms are defined: 
\begin{eqnarray*}
S_{d}&:=&M_{d\Omega}\ddot{q}_{d}+C_{d\Omega}\dot{q}_{d}
-M_{ed\Omega}-M_{vd\Omega}
+g_{d\Omega}+d_{d\Omega}
\\
M_{d\Omega}&:=&M_{\Omega}(q_{d},\dot{q}_{d}),
\ \
C_{d\Omega}:=C_{\Omega}(q_{d},\dot{q}_{d})
\\
M_{ed\Omega}&:=&M_{e\Omega}(q_{d},\dot{q}_{d}), 
\ \
M_{vd\Omega}:=M_{v\Omega}(q_{d},\dot{q}_{d})
\\
g_{d\Omega}&:=&g_{\Omega}(q_{d},\dot{q}_{d}), 
\ \
d_{d\Omega}:=d_{\Omega}(q_{d},\dot{q}_{d}). 
\end{eqnarray*}

To facilitate the stability analysis, 
the time derivative of (\ref{eqn:OpenDyn4}) 
can be determined as
\begin{eqnarray}
M_{\Omega}(q,\dot{q})\dot{r}
&=&-\dot{M}_{\Omega}(q,\dot{q})r+\dot{W}_{\Omega}
-\dot{u}
+\dot{d}_{\Omega}\nonumber\\
&=&-\dfrac{1}{2}\dot{M}_{\Omega}(q,\dot{q})r
+N
-\dot{u}
-e_{2}\nonumber\\
&=&-\dfrac{1}{2}\dot{M}_{\Omega}(q,\dot{q})r
+\tilde{N}+N_{d}
-\dot{u}
-e_{2} \label{eqn:OpenDyn5}
\end{eqnarray}
where $N$, $N_{d}$ and $\tilde{N}\in\R$ denote the following unmeasurable 
auxiliary terms
\begin{eqnarray*}
N&:=&
\dot{W}_{\Omega}
+e_{2}
-\dfrac{1}{2}\dot{M}_{\Omega}(q,\dot{q})r
+\dot{d}_{\Omega} \\
N_{d}&:=&\dot{S}_{d} \\
\tilde{N}&:=&N-N_{d}. 
\end{eqnarray*}
By applying the Mean Value Theorem, $\tilde{N}$ can be
upper bounded by state-dependent terms as
\begin{eqnarray}
\|\tilde{N}\|\leq\rho(\|z\|)\|z\| \label{eqn:UB_tildeN}
\end{eqnarray}
where $z\in\R^{3}$ is defined as
\begin{eqnarray}
 z:=\left[\begin{array}{ccc}
e_{1}& e_{2}& r
\end{array}
\right]^{T} \label{eqn:def_z}
\end{eqnarray} 
and $\rho(\|z\|)$ is some positive, 
nondecreasing function \cite{TAC04XiDaQuCh}.
By the design of the desired trajectory, $N_{d}$ can be 
upper bounded as 
\begin{eqnarray}
 \|N_{d}\|\leq \zeta_{N_{d}},\ \  \|\dot{N}_{d}\|\leq
  \zeta_{\dot{N}_{d}}, 
\label{eqn:zeta_Nd}
\end{eqnarray}
where $\zeta_{N_{d}}$, $\zeta_{\dot{N}_{d}}\in\R$ are
known positive constants. 

The voltage control input is designed as \cite{TNSRE09ShStGrDi}
\begin{eqnarray}
 u&=&(k_{s}+1)(e_{2}-e_{2}(0))+\nu  \label{eqn:u} \\
\dot{\nu}&=&(k_{s}+1)\alpha_{2}e_{2}+\beta{\sgn}(e_{2}),\ \nu(0)=\nu_{0}
\end{eqnarray}
where $\nu$ is the generalized Filippov solution to $\dot{\nu}$, 
$\nu_{0}$ is some initial condition, 
$k_{s}$, $\beta\in\R$ are positive, constant control gains, 
and ${\sgn}(\cdot)$ denotes the signum function. 

To facilitate the subsequent stability analysis, $y$ and $Q$ are 
defined as 
\begin{eqnarray}
 y:=\left[\begin{array}{c}
z\cr
\sqrt{P}
\end{array}
\right],\ \ 
Q:=\left[\begin{array}{ccc}
\alpha_{1}&-\frac{1}{2}&0\cr
-\frac{1}{2}&\alpha_{2}&0\cr
0&0&1\cr
\end{array}\right], 
\end{eqnarray}
where $P\in\R$ is the Filippov solution to
\begin{eqnarray}
 \dot{P}&=&-r^{T}\bigl(N_{d}-\beta\sgn(e_{2})
\bigr), 
\label{eqn:dotP(t)}\\
P(0)&=&\beta |e_{2}(0)|-e_{2}(0)N_{d}(0). \label{eqn:P(0)}
\end{eqnarray}
\begin{theorem}\label{thm:1}
The control law of (\ref{eqn:u}) yields semi-global 
asymptotic tracking in the sense that
\begin{eqnarray}
|e_{1}|\ \rightarrow\ 0\ \ as\ \ t\  \rightarrow\ \infty
\label{eqn:e1_ra_infty}
\end{eqnarray}
for the region of attraction ${\cal D}_{z}$
\begin{eqnarray}
{\cal D}_{z}=\left\{y\ |\ \rho\left(
\sqrt{\dfrac{\lambda_{2}}{\lambda_{1}}}\|y\|
\right)<2\sqrt{\lambda_{\min}(Q)k_{s}}\right\}
\end{eqnarray}
where 
$\lambda_{1}:=\frac{1}{2}\min\{1,\underline{M}_{\Omega}\}$, 
$\lambda_{2}:=\max\left\{\frac{1}{2}\overline{M}_{\Omega},1
\right\}$, 
and $\lambda_{\min}(Q)$ denotes the minimum eigenvalue of $Q$,  
provided $\alpha_{1}$, $\alpha_{2}$, $\beta$ and $k_{s}$ are selected 
according 
to the following sufficient conditions:
\begin{eqnarray}
\alpha_{1}\alpha_{2}&>&\dfrac{1}{4},\\
\beta&>&\left(\zeta_{N_{d}}+\dfrac{1}{\alpha_{2}}
\zeta_{\dot{N}_{d}}\right)\label{eqn:sufficient_condition},\\
k_{s}&>&\frac{1}{4\lambda_{\min}(Q)}\rho(\|z(0)\|)^{2},
\end{eqnarray}
where $\zeta_{N_{d}}$ and $\zeta_{\dot{N}_{d}}$ were introduced 
in (\ref{eqn:zeta_Nd}). 
\end{theorem}
\begin{proof}
The proof for Theorem~\ref{thm:1} closely follows the proof 
given in \cite{CDC13DoChDi}. 
The proof details are provided in Appendix C. 
\end{proof}

\section{CONCLUSIONS}
This paper considered tracking control for FES-cycling 
based on force direction efficiency derived from 
using the antagonistic bi-articular muscles. 
A muscle group force decomposition is developed 
to improve cycling efficiency. 
A tracking controller and associated stability analysis 
are developed for an uncertain nonlinear cycle-rider system in the 
presence of an unknown time-varying disturbance, and semi-global 
asymptotic tracking of the desired trajectories is guaranteed 
provided sufficient control gain conditions are satisfied. 
Ongoing efforts are focused on experimental demonstration 
of the developed controller.


\section*{APPENDIX}
\subsection{Reduced Model}\label{app:model}
Using the constraints in (\ref{eqn:constraints})
and the parameterization in (\ref{eqn:parameterization}), let 
\begin{eqnarray}
\psi(q')&:=&
\left[\begin{array}{c}
\phi(q')\cr
\alpha(q')\cr
\end{array}\right]
=\left[\begin{array}{c}
0\cr
q
\end{array}\right]. \label{eqn:psi(q')}
\end{eqnarray}
Differentiating (\ref{eqn:psi(q')}) with respect to time yields
\begin{eqnarray}
\psi_{q'}(q')\dot{q}'&=&\left[\begin{array}{ccc}
0&0&1
\end{array}\right]^{T}\dot{q}
\end{eqnarray}
where
\begin{eqnarray*}
\psi_{q'}(q')&:=&
\frac{\partial \psi(q')}{\partial q'}=\left[\begin{array}{c@{~}c@{~}c}
-l_{1}S_{1}-l_{2}S_{12}&
-l_{2}S_{12}&
l_{3}S_{3}\cr
l_{1}C_{1}+l_{2}C_{12}&
l_{2}C_{12}&
-l_{3}C_{3}\cr
0&0&1
\end{array}\right]\label{eqn:psi_{q'}(q')}. 
\end{eqnarray*}
Therefore $\mu(q')$ is obtained as
\begin{eqnarray}
\mu(q')=\psi_{q'}^{-1}(q')
\left[\begin{array}{ccc}
0&0&1
\end{array}\right]^{T}, \label{eqn:rho_q'}
\end{eqnarray}
where $\det(\psi_{q'})=l_{1}l_{2}S_{2}\neq 0$ except for 
$q_{2}=n\pi,\ n\in {\cal Z}$. 
Thus, there exists $\psi_{q'}^{-1}(q')$ by Assumption~\ref{ass:1}
i.e., the knee joint angle $q_{2}$ never equals $n\pi,\ n\in {\cal Z}$.  

By solving the constraints ${\cal C}$ in (\ref{eqn:constraints}),  
$q_{1}$ and $q_{2}$ can be represented 
as functions of $q_{3}$ as 
\begin{eqnarray}
q_{1}&=&\cos^{-1}\left(
\frac{l_{1}^{2}+(l_{3}C_{3}+c_{x})^{2}+(l_{3}S_{3}+c_{y})^{2}-l_{2}^{2}}
{2l_{1}\sqrt{(l_{3}C_{3}+c_{x})^{2}+(l_{3}S_{3}+c_{y})^{2}}}
\right)\nonumber\\
&&\ \ \ +\tan^{-1}\left(\frac{l_{3}S_{3}+c_{y}}
{l_{3}C_{3}+c_{x}}\right) \label{eqn:sigma_q1} \\
q_{2}&=&\cos^{-1}\left(\frac{l_{1}^{2}+l_{2}^{2}-
(l_{3}C_{3}+c_{x})^{2}-(l_{3}S_{3}+c_{y})^{2}}
{2l_{1}l_{2}}
\right)\nonumber\\
&&\ \ \ +\pi. \label{eqn:sigma_q2}
\end{eqnarray}
The expressions in (\ref{eqn:sigma_q1}) and (\ref{eqn:sigma_q2}) 
yield the parameterization $\sigma(q)$. 

\subsection{Analytic Solution of $q_{3}$ for 
$\chi_{{\rm Glut}}=1$ and $\chi_{{\rm Gast}}=1$}
This appendix develops on 
analytic solutions of $q_{3}$ at $\chi_{{\rm Glut}}=1$ and 
$\chi_{{\rm Gast}}=1$. 
The crank angle which satisfies that 
$\vec{F}_{\rm Glut}$ and $\vec{l}_{3}$ cross at right angles 
is denoted by $q_{\rm Glut}$. 
In other words, $q_{\rm Glut}$ equals $q_{3}$ which satisfies 
\begin{eqnarray}
q_{3}-\frac{\pi}{2}=q_{1}+q_{2}. \label{eqn:q3forGlut}
\end{eqnarray}
From (\ref{eqn:constraints}) and (\ref{eqn:q3forGlut}), 
\begin{eqnarray}
q_{\rm Glut}&=&\sin^{-1}\left(
\dfrac{l_{3}^{2}+l_{2}^{2}-l_{1}^{2}+c_{x}^{2}+c_{y}^{2}}
{-2\sqrt{(c_{y}l_{3}-c_{x}l_{2})^{2}+(c_{x}l_{3}+c_{y}l_{2})^{2}}}
\right)\nonumber\\
&&-\varphi_{1}+2n\pi,\ \ n\in{\cal Z}, 
\end{eqnarray}
where
\begin{eqnarray}
\varphi_{1}&:=&
\tan^{-1}\left(\dfrac{c_{x}l_{3}+c_{y}l_{2}}{c_{y}l_{3}-c_{x}l_{2}}
\right)+\pi. 
\end{eqnarray}

In a similar way, 
$q_{\rm Gast}$ is defined as a crank angle when 
$\vec{F}_{\rm Gast}$ and $\vec{l}_{3}$ cross at right angles. 
In other words, $q_{\rm Gast}$ equals $q_{3}$ which satisfies 
\begin{eqnarray}
q_{3}-\frac{\pi}{2}&=&\tan^{-1}\left(
\dfrac{l_{1}S_{1}+l_{2}S_{12}}{l_{1}C_{1}+l_{2}C_{12}}
\right)
=\tan^{-1}\left(
\dfrac{l_{3}S_{3}+c_{y}}{l_{3}C_{3}+c_{x}}
\right),\nonumber\\ \label{eqn:q3forGast}
\end{eqnarray}
where (\ref{eqn:constraints}) was utilized.  
From (\ref{eqn:q3forGast}), 
\begin{eqnarray}
\hspace*{-3ex}
q_{\rm Gast}=\sin^{-1}\left(\frac{l_{3}}{-\sqrt{c_{y}^{2}+c_{x}^{2}}}\right)
-\varphi_{2}+2n\pi,\ \ n\in{\cal Z}, 
\end{eqnarray}
where
\begin{eqnarray}
\varphi_{2}:=\tan^{-1}\left(\frac{c_{x}}{c_{y}}\right)+\pi. 
\end{eqnarray}

\subsection{Proof of Theorem \ref{thm:1}}
\begin{proof}
Consider the following positive definite 
continuously differentiable function
\begin{eqnarray}
V(y)=\frac{1}{2}r^{T}M_{\Omega}r+\frac{1}{2}e_{1}^{T}e_{1}
+\frac{1}{2}e_{2}^{T}e_{2}+P \label{eqn:V(y)}. 
\end{eqnarray}
Integrating (\ref{eqn:dotP(t)}) indicates that
\begin{eqnarray}
&&\hspace{-3ex}
P(t)-P(0)
\nonumber\\
&=&
-\int_{0}^{t}\alpha_{2} e_{2}(\tau)
\bigl(N_{d}(\tau)-\beta{\sgn}(e_{2}(\tau))\bigr)d\tau 
\nonumber\\
&&
-\int_{0}^{t}\dfrac{d(e_{2}(\tau))}{d\tau}
\bigl(
N_{d}(\tau)-
\beta{\sgn}(e_{2}(\tau))
\bigr)
d\tau
\nonumber\\
&=&
-\int_{0}^{t}\alpha_{2} e_{2}(\tau)
\bigl(N_{d}(\tau)-\beta{\sgn}(e_{2}(\tau))\bigr)d\tau 
\nonumber\\
&&
-e_{2}(\tau)
\left.\bigl.N_{d}(\tau)\right|_{0}^{t}
-\int_{0}^{t}e_{2}(\tau)
\dfrac{dN_{d}(\tau)}{d\tau}
d\tau
+\beta\left.\bigl.|e_{2}(\tau)|\right|_{0}^{t}
\nonumber\\
&=&
-\int_{0}^{t}\alpha_{2} e_{2}(\tau)
\left(N_{d}(\tau)
+\dfrac{1}{\alpha_{2}}\dfrac{dN_{d}(\tau)}{d\tau}
-\beta{\sgn}(e_{2}(\tau))\right)
d\tau 
\nonumber\\
&&
-e_{2}(t)N_{d}(t)
+e_{2}(0)N_{d}(0)
+\beta|e_{2}(t)|
-\beta|e_{2}(0)|
\nonumber\\
&=&
\int_{0}^{t}\alpha_{2} e_{2}(\tau)
\left(
\beta{\sgn}(e_{2}(\tau))
-N_{d}(\tau)
-\dfrac{1}{\alpha_{2}}\dfrac{dN_{d}(\tau)}{d\tau}
\right)
d\tau 
\nonumber\\
&&
-e_{2}(t)N_{d}(t)
+e_{2}(0)N_{d}(0)
+\beta|e_{2}(t)|
-\beta|e_{2}(0)|
\nonumber\\
&\geq&
\int_{0}^{t}\alpha_{2} |e_{2}(\tau)|
\left(
\beta
-|N_{d}(\tau)|
-\dfrac{1}{\alpha_{2}}\left|\dfrac{dN_{d}(\tau)}{d\tau}\right|
\right)
d\tau 
\nonumber\\
&&
+|e_{2}(t)|\bigl(\beta-|N_{d}(t)|\bigr)
-\bigl(\beta|e_{2}(0)|
-e_{2}(0)N_{d}(0)\bigr) \label{eqn:Pt-P0}
\end{eqnarray}
Based on the sufficient condition in (\ref{eqn:sufficient_condition}), 
(\ref{eqn:P(0)}) and (\ref{eqn:Pt-P0}) indicate that $P(t)\geq 0$,   
and (\ref{eqn:V(y)}) satisfies the following inequalities: 
\begin{eqnarray}
\lambda_{1}\|y\|^{2}\leq V\leq \lambda_{2}\|y\|^{2}. \label{eqn:ineq_V}
\end{eqnarray}

The time derivative of (\ref{eqn:V(y)}) exists almost 
everywhere (a.e.), i.e., for almost all $t\in[0,\infty)$, and 
$\dot{V}\overset{a.e.}\in \dot{\tilde{V}}$ where
\begin{eqnarray}
\dot{\tilde{V}}:=\underset{\xi\in \partial V}\cap
\xi^{T}K\left[
\dot{e}_{1}^{T}\ \ 
\dot{e}_{2}^{T}\ \ 
\dot{r}^{T}\ \ 
\frac{1}{2}P^{-\frac{1}{2}}\dot{P}\ \ 
1
\right]^{T} \label{eqn:def_dot_tilde_V}
\end{eqnarray}
and $\partial V$ is the generalized gradient of $V$. Since $V$ is
continuously differentiable, (\ref{eqn:def_dot_tilde_V}) can be 
rewritten as
\begin{eqnarray}
\dot{\tilde{V}}\subset\nabla V^{T}
K\left[\dot{e}_{1}^{T}\ \ 
\dot{e}_{2}^{T}\ \  
\dot{r}^{T}\ \  
\frac{1}{2}P^{-\frac{1}{2}}\dot{P}\ \ 
1 
\right]^{T} \label{eqn:dot_tilde_V1}
\end{eqnarray}
where $\nabla V:=\left[e_{1}^{T}\ \ e_{2}^{T}\ \ r^{T}M_{\Omega}\ \ 
2P^{\frac{1}{2}}\ \ \dfrac{1}{2}r^{T}\dot{M}_{\Omega}r
\right]^{T}$. 
Using $K[\cdot]$ from \cite{TCS87PaSa}, (\ref{eqn:dot_tilde_V1}) yields
\begin{eqnarray}
\dot{\tilde{V}}&\subset&
e_{1}^{T}(e_{2}-\alpha_{1}e_{1})
+e_{2}^{T}(r-\alpha_{2}e_{2})
\nonumber\\
&&
+r^{T}\Bigl(
-\dfrac{1}{2}\dot{M}_{\Omega}(q)r
+\tilde{N}+N_{d}
-(k_{s}+1)\dot{e}_{2} 
\nonumber\\
&&
-(k_{s}+1)\alpha_{2}e_{2}(t)-\beta K\bigl[{\sgn}(e_{2})\bigr]
-e_{2}\Bigr)
\nonumber\\
&&
+K\bigl[\dot{P}\bigr]
+\dfrac{1}{2}r^{T}\dot{M}_{\Omega}r.\label{eqn:dot_tilde_V2}
\end{eqnarray}
By substituting $\dot{P}$ from (\ref{eqn:dotP(t)}), 
(\ref{eqn:dot_tilde_V2}) can be transformed into
\begin{eqnarray}
\dot{\tilde{V}}&\subset&
e_{1}^{T}(e_{2}-\alpha_{1}e_{1})
-\alpha_{2}e_{2}^{2}
\nonumber\\
&&
+r^{T}\Bigl(
\tilde{N}+N_{d}
-(k_{s}+1)r
-\beta K\bigl[{\sgn}(e_{2})\bigr]
\Bigr)
\nonumber\\
&&
+K\bigl[-r^{T}(N_{d}-\beta{\sgn}(e_{2}))\bigr]
\nonumber\\
&=&
e_{1}^{T}(e_{2}-\alpha_{1}e_{1})
-\alpha_{2}e_{2}^{2}
\nonumber\\
&&
+r^{T}\Bigl(
\tilde{N}
-(k_{s}+1)r
-\beta K\bigl[{\sgn}(e_{2})\bigr]
\Bigr)
\nonumber\\
&&
+r^{T}\beta K\bigl[\sgn(e_{2})\bigr]. 
\label{eqn:dot_tilde_V3}
\end{eqnarray}
Eq.~(\ref{eqn:dot_tilde_V3}) can be further upper bounded as
\begin{eqnarray}
\dot{V}&\overset{a.e.}\leq&
-\alpha_{1}e_{1}^{2}
+e_{1}^{T}e_{2}
-\alpha_{2}e_{2}^{2}
+r^{T}\tilde{N}
-(k_{s}+1)r^{2}
\nonumber\\
&=&r^{T}\tilde{N}-k_{s}r^{2}-z^{T}Qz
\label{eqn:dot_tilde_V4}
\end{eqnarray}
where the set in (\ref{eqn:dot_tilde_V3}) reduces to the 
scalar inequality in (\ref{eqn:dot_tilde_V4}) because 
the right-hand side is continuous a.e., i.e., 
the right-hand side is continuous except for the Lebesgue negligible 
set of times when
\begin{eqnarray*}
r^{T}\Bigl(
\beta K[\sgn(e_{2})]-\beta K[\sgn(e_{2})]
\Bigr) \neq \{0\}. 
\end{eqnarray*}
By using Eq.~({\ref{eqn:UB_tildeN}}), 
the term $r^{T}\tilde{N}$ can be upper bounded as 
\begin{eqnarray}
\|r^{T}\tilde{N}\|\leq\rho(\|z\|)\|z\|\|r\|
\end{eqnarray}
to obtain
\begin{eqnarray}
\dot{V}\overset{a.e.}\leq-\lambda_{\min}(Q)\|z\|^{2}
+\rho(\|z\|)\|z\|\|r\|-k_{s}\|r\|^{2}.
\end{eqnarray}
By completing the squares, 
\begin{eqnarray}
\dot{V}&\overset{a.e.}\leq&-\lambda_{\min}(Q)\|z\|^{2}
-k_{s}\left(
\|r\|-\dfrac{\rho(\|z\|)\|z\|}{2k_{s}}
\right)^{2}\nonumber\\
&&\ \ \ +\dfrac{\rho(\|z\|)^{2}\|z\|^{2}}{4k_{s}}\nonumber\\
&\leq&
-\left(\lambda_{\min}(Q)-\dfrac{\rho(\|z\|)^{2}}{4k_{s}}\right)\|z\|^{2}. 
\label{eqn:dotV}
\end{eqnarray}
From (\ref{eqn:dotV}), it follows that 
\begin{eqnarray}
\dot{V}\overset{a.e.}\leq
-U = -\gamma \|z\|^{2},\ \  \forall y\in D \label{eqn:ineq_dotV}
\end{eqnarray}
where $\gamma\in\R$ is some positive constant, and 
$D:=\{y\in \R^{3+1}\ |\ \rho(\|y\|)<2\sqrt{\lambda_{\rm min}(Q)
k_{s}}\}$. From the inequalities in (\ref{eqn:ineq_V}) and  
(\ref{eqn:ineq_dotV}), $V\in {\cal L}_{\infty}$, and hence, 
$e_{1}$, $e_{2}$, and $r\in L_{\infty}$. The remaining signals in 
the closed-loop dynamics can be proven to be bounded. By the 
definition of $z$ in (\ref{eqn:def_z}) and $U$ in 
(\ref{eqn:ineq_dotV}), $U$ can be shown to be uniformly continuous. 
Then, the region of attraction $D_{z}$ can be expanded arbitrarily 
by increasing $k_{s}$. By invoking Corollary 1 in \cite{TAC13FiKaDi}, 
$\gamma\|z\|^{2}\rightarrow 0$ as $t\rightarrow\infty$, 
$\forall y(0)\in {\cal D}_{z}$. Based on the definition of $z$, 
$e_{1}\rightarrow 0$ as $t\rightarrow\infty$, 
$\forall y(0)\in {\cal D}_{z}$. 
\end{proof}

\end{document}